\documentclass[11pt,letterpaper]{amsart}

\usepackage{amsmath,amsfonts,amsthm,amssymb,stmaryrd,bm,cite,enumerate}
\usepackage{relsize,tabls}

\theoremstyle{plain}

\numberwithin{equation}{section}

\newtheorem{thm}{Theorem}[section]
\newtheorem{cor}[thm]{Corollary}

\theoremstyle{definition}
\newtheorem{example}{Example}

\allowdisplaybreaks  

\newcounter{cond}

\newcommand{\Natural}{{\mathbb N}}
\newcommand{\real}{{\mathbb R}}

\newcommand{\tbullet}{\mathrel{\raise .2ex\hbox{\tiny$\bullet$}}} 

\newcommand{\rmob}{\mathrm{Ob\,}}
\newcommand{\rmco}{\mathrm{co\,}}
\newcommand{\rms}{\mathrm{s\,}}
\newcommand{\ityes}{\textit{yes}}      
\newcommand{\itno}{\textit{no}}

\newcommand{\sscript}{\mathcal{S}}

\newcommand{\ab}[1]{\left|#1\right|}
\newcommand{\brac}[1]{\left\{#1\right\}}
\newcommand{\paren}[1]{\left(#1\right)}
\newcommand{\sqbrac}[1]{\left[#1\right]}

\begin{document}

\title{FINITE CLASSICAL AND QUANTUM EFFECT ALGEBRAS}

\author{Stan Gudder}
\address{Department of Mathematics, 
University of Denver, Denver, Colorado 80208}
\email{sgudder@du.edu}
\date{}
\maketitle

\begin{abstract}
In this article, we only consider finite effect algebras. We define the concepts of classical and quantum effect algebras and show that an effect algebra $E$ is classical if and only if there exists an observable that measures every effect of $E$. We next consider matrix representations of effect algebras and prove an effect algebra is classical if and only if its matrix representation has precisely one row. We then discuss sum table for effect algebras. Although these are not as concise as matrix representations, they give more immediate information about effect sums which are the basic operations of an effect algebra. We subsequently study states on effect algebras and prove that classical effect algebras are quantum effect algebras. Finally, we consider composites of effect algebras. This allows us to study interacting systems described by effect algebras. We show that two effect algebras are classical if and only if their composite is classical. We point out that scale effect algebras are not the only classical effect algebras and stress the importance of atoms in this work.
\end{abstract}

\section{Introduction}  
The basic entity in quantum measurement theory is an effect \cite{bgl95,blm96,dp00,fb94}. Effects correspond to two-valued, \ityes-\itno\ (true-false) experiments. If a measurement of an effect $a$ results in the outcome \ityes\ (true), we say that $a$ \textit{occurs} and if it results in the outcome \itno\ (false) we cay that $a$
\textit{does not occur}. There are two special effects $0,1$ which never occurs and always occurs, respectively. For every effect $a$ there is a 
\textit{complementary} effect $a'$ that occurs if and only if $a$ does not occur. Finally, certain pairs of effects $a,b$ can be measured in parallel which results in another effect denoted by $a\oplus b$. If this happens we say that $a\oplus b$ is \textit{defined} and otherwise $a\oplus b$ is \textit{not defined}. Rigorous definitions for these concepts will be given in Section~2 where we define an effect algebra. An \textit{observable} $A$ is a collection of effects
$A=\brac{A_x\colon x\in\Omega _A}$ where $\Omega _A$ is called the \textit{outcome space} of $A$ and $\bigoplus\limits _{x\in\Omega _A}A_x=1$ which means that some outcome $x\in\Omega _A$ always results \cite{bgl95,blm96,dp00}. We say that $A$ \textit{measures} the effects that are sums of $A_x$'s and if the result of a measurement of $A$ is outcome $x$, then the effect $A_x$ occurs. In Section~2, we study the concept of simultaneous observables and define classical and quantum effect algebras. From now on, unless stated otherwise, all effect algebras will be assumed to be finite \cite{bkz23}.

Section~3 considers matrix representations of effect algebras. This gives a concise way of describing an effect algebra without even listing its elements. We show that an effect algebra is classical if and only if its matrix representation has precisely one row. In Section~4 we discuss sum tables for effect algebras. Although these are of as concise as matrix representations, they give more immediate information about the sums $a\oplus b$ which are the basic operations of an effect algebra. Section~5 studies states on effect algebras \cite{gt24}. We prove the important result that classical effect algebras are quantum effect algebras. Of course, the converse does not hold and examples are given. Finally, in Section~6 we consider composites of effect algebras \cite{gt24}. This allows us to study interacting systems described by effect algebras. We show that effect algebras are classical if and only if their composite effect algebras is classical.

To be precise, we say that an effect algebra $E$ is \textit{classical} if any two observables on $E$ are simultaneous. Physically, this means that any two observables on $E$ do not interfere and can be measured together. We show that this is equivalent to the existence of an observable on $E$ that measures all the nonzero effects on $E$. Until now it has been thought that there is only one classical effect algebra (up to isomorphism) with $n$ elements. Such an effect algebra is called a \textit{scale} effect algebra \cite{dp00,fb94}. To show this is not true we exhibit two classical effect algebras with 4, 6, 9, 10, 14 and 15 elements, respectively. Moreover, we find three classical effect algebras with 8 elements, four with 12 elements, five with 16 elements and nine with 36 elements, respectively.

An example of a scale effect algebra is given by a millimeter ruler. Suppose we have a ruler one meter long that is divided into 1,000 millimeter hash marks. This ruler can distinguish lengths to within one millimeter and we call one millimeter the scale of the effect algebra. This scale effect algebra consists of the lengths
$E=\brac{0,\tfrac{1}{1,000}\,,\tfrac{2}{1,000}\,,\cdots ,1}$ that the ruler can measure and each length is an effect in $E$. The corresponding generating observable for $E$ is $A=\brac{\tfrac{1}{1,000}\,,\tfrac{1}{1,000}\,,\cdots ,\tfrac{1}{1,000}}$ where there are 1,000 terms. We see that $A$ is an observable because the sum of its terms is 1. In this case, each nonzero effect in $E$ is the sum of terms of $A$. Notice that if $a\in E$, then $a'=1-a\in E$. Examples of sums are
$\tfrac{3}{1,000}\oplus\tfrac{5}{1,000}=\tfrac{8}{1,000}$ and $\tfrac{10}{1,000}\oplus\tfrac{990}{1,000}=1$. Not all sums are defined, for example
$\tfrac{11}{1,000}\oplus\tfrac{990}{1,000}$ is not defined.

To be precise, a scale effect algebra with $n$ elements has the form
\begin{equation*}
E=\brac{0,a,a\oplus a,\ldots ,a\oplus a\oplus\cdots\oplus a=1}
\end{equation*}
where $a\oplus a\oplus\cdots\oplus a$ has $n-1$ $a$'s and $a$ is the scale for $E$. The simplest example of a classical effect algebra that is not a scale effect algebra is $E=\brac{0,a,b,1}$ with four elements where $a\oplus b=1$ and $a\oplus a$, $b\oplus b$, $a\oplus 1$, $b\oplus 1$ are not defined. We have that $E$ is classical because $A=\brac{a,b}$ is a generating observable. The other classical effect algebra with four elements is the scale effect algebra
$F=\brac{0,c,c\oplus c,1}$ where $c\oplus c\oplus c=1$. The generating observable is $B=\brac{c,c,c}$. It is clear that $E$ and $F$ are not isomorphic.

\section{Simultaneous Observables}  
An effect algebra is a set $E$ with two special elements $0,1\in E$ and a partial binary operation $\oplus$ satisfying the following conditions
\cite{dp00,fb94,gud97,gt24}:
\begin{list}
{(E\arabic{cond})}{\usecounter{cond}
\setlength{\rightmargin}{\leftmargin}}
\item if $a\oplus b$ is defined, then $b\oplus a$ is defined and $b\oplus a=a\oplus b$,
\item if $b\oplus c$, $a\oplus (b\oplus c)$ are defined, then $a\oplus b,(a\oplus b)\oplus c$ are defined and
$a\oplus (b\oplus c)=(a\oplus b)\oplus c$
\item for every $a\in E$, there exists a unique element $a'\in E$ such that $a\oplus a'=1$,
\item if $e\oplus 1$ is defined, then $e=0$.
\end{list}
An element of an effect algebra is called an \textit{effect} and the partial operation $\oplus$ is referred to as an effect sum or just a sum. Unless confusion threatens, we say that $E$ is an effect algebra, when we really mean $(E,\oplus ,0,1)$ is an effect algebra. If $a\oplus b$ is defined, we write $a\perp b$ and it follows from the axioms that $0\perp a$ for $a\in E$ and $0\oplus a=a$. Moreover, $0'=1$ and $a''=a$ for all $a\in E$. We write $a\le b$ if there exists a $c\in E$ such that $b=a\oplus c$. It is easy to show that $c$ is unique and we write $c=b\ominus a$. It follows that $\le$ is a partial order on $E$ and $0\le a\le 1$ for all $a\in E$. Also, one can verify that $a\perp b$ if and only if $a\le b'$. We now recall the usual standard examples of (infinite) effect algebras.

\begin{example}  
(Standard scale effect algebra) The interval $[0,1]\subseteq\real$ is an effect algebra where $a\perp b$ whenever $a+b\le 1$ and in this case $a\oplus b=a+b$. We also have, $a'=1-a$ and $a\le b$ is the usual order so $b\ominus a=b-a$.\hfill\qedsymbol
\end{example}

\begin{example}  
(Standard quantum effect algebra) Let $H$ be a Hilbert space and denote by $E(H)$ the set of all bounded linear operators $a$ on $H$ satisfying the operator inequalities $0\le a\le I$. We define $a\perp b$ whenever $a+b\le I$ and in this case $a\oplus b=a+b$. Again, $a'=I-a$ and $a\le b$ is the usual order so
$b\ominus a=b-a$.\hfill\qedsymbol
\end{example}

In our later investigation a special role will be on those effect algebras that are finite subalgebras of these two standard cases. The main difference between a finite and infinite effect algebra is that a finite effect algebra possesses a generating set of atoms while an infinite effect algebra usually does not. This fact forms the basis for much of the work in the present article. We now elaborate on this statement. An effect $a$ of an effect algebra $E$ is an \textit{atom} if $a\ne 0$ and if $b\in E$ with $b\le a$ implies $b=0$ or $b=a$. Thus, an atom is a smallest nonzero element of an effect algebra. We denote the cardinality of a set $S$ by $\ab{S}$. 

\begin{thm}    
\label{thm21}
If $E$ is a finite effect algebra, then any $b\in E$ with $b\ne 0$ has the form
\begin{equation}                
\label{eq21}
b=a_1\oplus a_2\oplus\cdots\oplus a_n
\end{equation}
where $a_i\in E$ are atoms, $i=1,2,\ldots ,n$ and $a_i=a_j$ is allowed.
\end{thm}
\begin{proof}
If $b$ is an atom, we are finished. Otherwise, there exists an $a\in E$ with $0<a<b$. If $a$ is an atom, let $a_1=a$. If $a$ is not an atom, let $c\in E$ with $0<c<a$. If $c$ is an atom, let $a_1=c$. If $c$ is not an atom, there exists $d\in E$ with $0<d<c$. This process must eventually stop because $\ab{E}<\infty$. We conclude that there exists an atom $a_1$ such that $b=a_1$ or $a_1<b$. If $a_1<b$ then there exists an $e\in E$ with $e\ne 0$ such that $b=a_1\oplus e$. Continuing, we have atoms $a_1,a_2,\ldots ,a_m$ such that
$b=a_1\oplus\cdots\oplus a_m\oplus e$. Since $\ab{E}<\infty$, this process must end so
\eqref{eq21} holds.
\end{proof}

A set of atoms $\brac{a_1,a_2,\ldots ,a_n}\subseteq E$ is \textit{generating} if every $b\in E$ with $b\ne 0$ has the form \eqref{eq21}. The point is that a generating set need not contain all the atoms of $E$ and it follows from Theorem \ref{thm21} that a generating set of atoms always exists. If $a\in E$ is an atom and $b=a\oplus a\oplus\cdots\oplus a$ ($n$ summands) is defined we write $b=na (0=0a)$. If $\brac{a_1,a_2,\ldots ,a_n}\subseteq E$ is the set of all atoms of $E$, then every nonzero $b\in E$ has the form 
\begin{equation}                
\label{eq22}
b=n_1a_1\oplus n_2a_2\oplus\cdots\oplus n_na_n
\end{equation}
Recall from Section 1, that an \textit{observable} on $E$ is a finite collection of effects $A=\brac{a_1,a_2,\ldots ,a_n}$ where $\bigoplus\limits _{i=1}^na_i=1$ and $a_i$ is the effect that occurs when a measurement $A$ results in the outcome $i$. We say that $A$ is \textit{atomic} if $a_i$ is an atom, $i=1,2,\ldots ,n$ and $a_i=a_j$ is allowed. We shall see that an atomic observable need not contain all the atoms as its effects. We denote the set of observables on $E$ by
$\rmob (E)$. For $A,B\in\rmob (E)$ we write $B\le A$ if every effect of $B$ is a sum of effects of $A$. It is clear that $A\le A$ and $C\le B\le A$ implies $C\le A$. However, $\le$ is not a partial order because $A\le B$, $B\le A$ need not imply $A=B$. For example, let $F=\brac{0,c,c\oplus c,1}$, where $c\oplus c\oplus c=1$, be the effect algebra considered in Section~1. Then the observables $A,B\in\rmob (F)$ given by $A=\brac{c,c,c}$, $B=\brac{c,c\oplus c}$ satisfy $A\le B$,
$B\le A$ but $A\ne B$.

\begin{cor}    
\label{cor22}
If $B\in\rmob (E)$, there exists an atomic $A\in\rmob (E)$ such that $B\le A$.
\end{cor}
\begin{proof}
If $B=\brac{b_1,b_2,\ldots ,b_m}$ and $\brac{a_1,a_2,\ldots ,a_n}$ is the set of atoms of $E$, then every $b_i$ has the form \eqref{eq22}. Letting $A$ be the observable
\begin{equation*}
A=\brac{a_1,\dots ,a_1,a_2\ldots ,a_2,\ldots ,a_n,\ldots ,a_n}
\end{equation*}
where the number of $a_i$'s is sufficiently large, we have $B\le A$.
\end{proof}

An observable $A=\brac{a_1,a_2,\ldots ,a_n}$ \textit{measures} an effect $b$ if $b=a_{i_1}\oplus a_{a_2}\oplus\cdots\oplus a_{i_m}$ for some $i_j\in\brac{1,2,\ldots ,n}$, $j=1,2,\ldots ,m$ \cite{gt24}. If $A$ measures $b$, we write $b\le A$ and if $A$ measures every $b\in E$ then $A$ is \textit{generating} for $E$. Two effects $a,b\in E$ are \textit{compatible} ($a\,\rmco b$) if there exists a $c\in E$ such that
$c\le a,b$ and
$(a\ominus c)\oplus (b\ominus c)\oplus c$ is defined. For example, if $a\perp b$, then we can let $c=0$ so $a\,\rmco b$. Also, if $a\le b$, then we can let $c=a$ so $a\,\rmco b$. We conclude that $0\,\rmco a$ and $1\,\rmco a$ for every effect $a$.

\begin{thm}    
\label{thm23}
Two effects $a,b\in E$ are compatible if and only if there exists an observable $A\in\rmob (E)$ such that $a,b\le A$.
\end{thm}
\begin{proof}
Suppose $a\,\rmco b$. Then there exists an effect $c$ such that $c\le a,b$ and $(a\ominus c)\oplus (b\ominus c)\oplus c$ is defined. It follows that there exists an effect $d$ such that 
\begin{equation*}
(a\ominus c)\oplus (b\ominus c)\oplus c\oplus d=1
\end{equation*}
Let $A\in\rmob (E)$ be given by $A=\brac{a\ominus c,b\ominus c,c,d}$. Then $a=(a\ominus c)\oplus c$ and $b=(b\ominus c)\oplus c$ so $a,b\le A$. Conversely, suppose $a,b\le A$ where $A\in\rmob (E)$ and $A=\brac{a_1,a_2,\ldots ,a_n}$. Then
$a=a_{i_1}\oplus a_{i_2}\oplus\cdots\oplus a_{i_s}$ and
$b=a_{j_1}\oplus a_{j_2}\oplus\cdots\oplus a_{j_t}$. If $a_{i_u}\ne a_{j_v}$ for all $u,v$, then $a\perp b$ so $a\,\rmco b$. If
\begin{equation*}
a_{i_1}\oplus a_{i_2}\oplus\cdots\oplus a_{i_s}=a_{j_1}\oplus a_{j_2}\oplus\cdots\oplus a_{j_t}=c
\end{equation*}
then $c\le a,b$ and $(a\ominus c)\oplus (b\ominus c)\oplus c$ is defined so $a\,\rmco b$.
\end{proof}

We conclude that two that two effects are compatible if and only if they are simultaneously measured by a single observable. Let
$A=\brac{A_x\colon x\in\Omega _a}$, $B=\brac{B_y\colon y\in\Omega _B}$ be observables on $E$. We say that $A,B$ are \textit{simultaneous} $(A\,\rms B)$ if there exists a $C\in\rmob (E)$ such that $A,B\le C$. We also say that $A,B$ are \textit{simultaneously measurable}. If $A\,\rms B$, then $A_x\,\rmco B_y$ for all $x\in\Omega _A$, $y\in\Omega _B$. Indeed, there exists $C\in\rmob (E)$ such that $A,B\le C$ so $C$ measures $A_x,B_y$ and by Theorem~\ref{thm23},
$A_x\rmco B_y$. The converse does not hold \cite{bgl95,blm96,dp00}.

An effect algebra $E$ is \textit{classical} if $A\,\rms B$ for all $A,B\in\rmob (E)$.

\begin{thm}    
\label{thm24}
The following statements are equivalent.
{\rm{(i)}}\enspace $E$ is classical.
{\rm{(ii)}}\enspace There exists an $A\in\rmob (E)$ such that $B\le A$ for every $B\in\rmob (E)$.
{\rm{(iii)}}\enspace  There exists an $A\in\rmob (E)$ such that $b\le A$ for all $b\in E$.
\end{thm}
\begin{proof}
(i)$\Rightarrow$(ii)\enspace Suppose $E$ is classical. Since $E$ is finite, there exists a finite number of observables $A_1,A_2,\ldots ,A_n\in\rmob (E)$. Since
$A_1\,\rms A_2$ there exists a $C_{12}\in\rmob (E)$ such that $A_1,A_2\le C_{12}$. Since $A_3\,\rms C_{12}$ there exists a $C_{123}$ such that 
$A_3,C_{12}\le C_{123}$. Therefore, $A_1,A_2,A_3\le C_{123}$. Continuing, there exists a $C\in\rmob (E)$ such that $A_1,A_2,\ldots ,A_n\le C$.
(ii)$\Rightarrow$(iii)\enspace Suppose there exists $A\in\rmob (E)$ such that $B\le A$ for all $B\in\rmob (E)$. If $b\in E$ with $b\ne 0,1$, then
$B=\brac{b,b'}\in\rmob (E)$. Since $B\le A$ we have $b\le B\le A$.
(iii)$\Rightarrow$(i)\enspace Suppose there is an $A\in\rmob (E)$ with $b\le A$ for all $b\in E$. If $B\in\rmob (E)$ and $b\le B$, then $b\le A$ so $B\le A$. If
$C\in\rmob (E)$, then $C\le A$ so $B\,\rms C$. Hence, $E$ is classical.
\end{proof}

Notice that $0\le A$ for all $A\in\rmob (E)$ because $0=\bigoplus\limits _{x\in\emptyset}A_x$. Also, $1\le A$ for all $A\in\rmob (E)$. If $E,F$ are effect algebras, a map $f\colon E\to F$ is \textit{additive} if $a,b\in E$ with $a\perp b$ implies $f(a)\perp f(b)$ and $f(a\oplus b)=f(a)\oplus f(b)$. An additive map that satisfies $f(1)=1$ is a \textit{morphism} \cite{fb94}. A morphism $f\colon E\to F$ is a \textit{monomorphism} if $f(a)\perp f(b)$ implies $a\perp b$ and
$f(a\oplus b)=f(a)\oplus f(b)$. A surjective monomorphism is called an \textit{isomorphism} \cite{fb94}. We frequently identify isomorphic effect algebras. We say that $E$ is an \textit{effect subalgebra} of $F$ if there exists a monomorphism $f\colon E\to F$. A finite effect algebra $E$ is a \textit{scale effect algebra} if $E$ is an effect subalgebra of the standard scale effect algebra $[0,1]$ of Example~1. A finite effect algebra $E$ is a \textit{quantum effect algebra} if $E$ is an effect subalgebra of the standard quantum effect algebra $E(H)$ of Example~2 \cite{gt24}. An effect algebra $E$ is a \textit{real-model} of an effect algebra $F$ if
$E=\brac{0,\lambda _1,\lambda _2,\ldots ,\lambda _n,1}$ $\lambda _i\in (0,1)\subseteq\real$, where $\lambda _i\perp \lambda _j$ if and only if $\lambda _i+\lambda _j\le 1$ and $E$ is isomorphic to $F$. If we only have $\lambda _i\perp\lambda _j$ when
$\lambda _i+\lambda _j\in E$ and $E$ is isomorphic to $F$, then $E$ is a \textit{weak real-model} of $F$. It is clear that a real-model is a weak real model but the converse does not hold.

\begin{example}  
We now show that a scale effect algebra $F$ is classical. If $F$ has $n+1$ elements, then the effect algebra
$E=\brac{0,\tfrac{1}{n}\,,\tfrac{2}{n}\,,\cdots ,\tfrac{n-1}{n}\,,1}$ is a real-model for $F$. Now $A=\brac{\tfrac{1}{n}\,,\tfrac{1}{n}\,,\cdots ,\tfrac{1}{n}}$ where there are $n$ terms is an observable on $E$. Since $\tfrac{m}{n}\le A$, $m=0,1,\ldots ,n$, it follows from Theorem~\ref{thm24} that $E$ is claasical. Now $F$ is isomorphic to $E$ so $F$ is classical.\hfill\qedsymbol
\end{example}

\begin{example}  
An example of a classical effect algebra that is not a scale effect algebra is given at the end of Section~1. We now give another example. Let
\begin{equation*}
E=\brac{0,\tfrac{1}{5}\,,\tfrac{1}{3}\,,\tfrac{7}{15}\,,\tfrac{8}{15}\,,\tfrac{2}{3}\,,\tfrac{4}{5}\,,1}
\end{equation*}
be a weak real-model for an effect algebra where $\lambda _i\perp\lambda _j$ if $\lambda _i+\lambda _j\in E$. Since
$A=\brac{\tfrac{1}{5}\,,\tfrac{1}{3}\,,\tfrac{7}{15}}\in\rmob (E)$ and every $a\in E$ satisfies $a\le A$, it follows from Theorem~\ref{thm24} that $E$ is classical. However, it is clear that $E$ is not a scale effect algebra.\hfill\qedsymbol
\end{example}

\begin{example}  
We now show there are effect algebras that are not classical. The weak real-model $\brac{0,\tfrac{1}{5}\,,\tfrac{1}{2}\,,\tfrac{4}{5}\,,1}$ has observables,
$\brac{1},\brac{\tfrac{1}{2}\,,\tfrac{1}{2}},\brac{\tfrac{1}{5}\,,\tfrac{4}{5}},$ $\brac{\tfrac{1}{5}\,,\tfrac{1}{5}\,,\tfrac{1}{5}\,,\tfrac{1}{5}\,,\tfrac{1}{5}}$ and none of these observables have all effects below them. By Theorem~\ref{thm24} this effect algebra is not classical. Another example is
$\brac{0,\tfrac{1}{4}\,,\tfrac{1}{3}\,,\tfrac{2}{3}\,,\tfrac{3}{4}\,,1}$ which has observables
$\brac{1},\brac{\tfrac{1}{3}\,,\tfrac{2}{3}},\brac{\tfrac{1}{4}\,,\tfrac{3}{4}}\brac{\tfrac{1}{3}\,,\tfrac{1}{3}\,,\tfrac{1}{3}}$, and
$\brac{\tfrac{1}{4}\,,\tfrac{1}{4}\,,\tfrac{1}{4}\,,\tfrac{1}{4}}$. None of these observables have all effects below them.\hfill\qedsymbol
\end{example}

It is shown in later sections that there are quantum effect algebras that are not classical. If an effect $a$ is an atom we denote by $n(a)$ the largest integer such that $n(a)a$ is defined. We end this section with another characterization of classical effect algebras.

\begin{thm}    
\label{thm25}
An effect algebra $E$ is classical if and only if
\begin{equation}                
\label{eq23}
n(a_1)a_1\oplus n(a_2)a_2\oplus\cdots\oplus n(a_m)a_m=1
\end{equation}
where $a_i$ are the atoms of $E$, $i=1,2,\ldots ,m$.
\end{thm}
\begin{proof}
Suppose $E$ is classical. Then by Theorem~\ref{thm24}, there exists an observable $A\in\rmob (E)$ such that $b\le A$ for every $b\in E$. We now show that 
\begin{equation}                
\label{eq24}
A=\brac{a_1,\ldots ,a_1,a_2,\ldots ,a_2,\ldots ,a _m,\ldots ,a_m}
\end{equation}
where there are $n(a_i)$ copies of $a_i$, $i=1,2,\ldots ,m$. We know that $n(a_i)\ge 1$ because $a_i\le A$, $i=1,2,\ldots ,m$. If
$a_1\perp a_1$, then
$a_1\oplus a_1=\oplus b_i$ where $b_i\in A$. Hence,
\begin{equation*}
a_1=a_1\oplus a_1\ominus a_1=\oplus b_i\ominus a_1\in A
\end{equation*}
We conclude there are two versions of $a_1$ in $A$ so $\brac{a_1,a_1}\subset A$. Suppose $a_1\oplus a_1\oplus a_1$ is defined. Then
$\brac{a_1,a_1}\subseteq A$ as before and $a_1\oplus (a_1\oplus a_1)=\oplus b_i$ where $b_i\in A$. As before, $\brac{a_1,a_1,a_1}\subseteq A$.
Continuing, $\brac{a_1,a_1,\ldots , a_1}\subseteq A$ where there are $n(a_1)$ terms. The same holds for $a_2,a_3,\ldots a_n$. We conclude that 
\begin{equation*}
n(a_1)a_1\oplus n(a_2)a_2\oplus\cdots\oplus n(a_m)a_m\le 1
\end{equation*}
If $n(a_1)a_1\oplus n(a_2)a_2\oplus\cdots\oplus n(a_m)a_m<1$, then there exists an atom $a_i$ such that 
\begin{equation*}
n(a_1)a_1\oplus  n(a_2)a_2\oplus\cdots\oplus n(a_m)a_m\oplus a_i
\end{equation*}
is defined. Buth then $\paren{n(a_i)+1}a_i$ is defined which is a contradiction. Hence, \eqref{eq23} holds. Conversely, suppose \eqref{eq23} holds and $A$ is the observable \eqref{eq24}. If $b\in E$, then by \eqref{eq22} we have $b=\oplus x_ia_i$, $x_i\in\brac{0,1,\ldots ,n(a_i)}$. Hence $b\le A$ so $E$ is classical.
\end{proof}

It follows from Theorem~\ref{thm25} that if $E$ is classical, then the observable $A$ given by \eqref{eq24} is generating for $E$. We conclude that $E$ is classical if and only if $E$ possesses a generating, atomic observable.

\section{Matrix Representations}  
Let $E$ be a finite effect algebra with $m$ atoms $a_1,a_2,\ldots ,a_m$ and let $\Natural _0$ be the set of nonnegative integers. If
$\bigoplus\limits _{\lambda =1}^mx_ia_i=1$, $x_i\in\Natural _0$, we call $\brac{x_1a_1,x_2a_2,\ldots ,x_ma_m}$ a \textit{semi-atomic observable} or an
\textit{atomic row} and $[x_1x_2\ldots x_m]$ a \textit{row} of $E$. If there are $n$ different rows $[y_{11}y_{12}\cdots y_{1m}]$, \newline 
$[y_{21}y_{22}\cdots y_{2m}],\ldots ,[y_{n1}y_{n2}\cdots y_{nm}]$ for $E$, we call the $n\times m$ matrix $M(E)=[y_{ij}]_{n\times m}$ a
\textit{matrix representation} of $E$ \cite{bkz23}. We consider an interchange of rows or columns of $M(E)$ to be equivalent and write $M(E)\approx M(F)$ if $M(E)$ can be obtained from $M(F)$ by an interchange of rows or columns of $M(F)$. The following result is proved in \cite{bkz23}.

\begin{thm}    
\label{thm31}
$E$ is isomorphic to $F$ if and only if $M(E)\approx M(F)$.
\end{thm}

Necessary and sufficient conditions for an $n\times m$ matrix $[y_{ij}]_{n\times m}$ where $y_{ij}\in\Natural _0$ to be a matrix representation for an effect algebra with $m$ atoms are given in \cite{bkz23}.

\begin{thm}    
\label{thm32}
An $n\times m$ matrix $M=[y_{ij}]_{n\times m}$, $y_{ij}\in\Natural _0$ is a matrix representation for an effect algebra with $m$ atoms if and only if\newline
{\rm{(1)}}\enspace All rows and columns of $M$ have at least one nonzero entry.\newline
{\rm{(2)}}\enspace Let $e_k=[0\cdots 010\cdots 0]$ where 1 is the $k$th entry. If $r_i$ is the $i$th row and $r_j$ is the $j$th row and
$r_j\ge r_i-e_k\ge 0$ for some $1\le k\le m$, then $r_i=r_j$.\newline
{\rm{(3)}}\enspace If $r_i,r_j,r_k$ are rows of $M$ and $r_i+r_j\ge r_k$, then $r_i+r_j-r_k$ is a row of $M$.
\end{thm}

It follows from (2) that distinct rows of $M$ are different.

\begin{thm}    
\label{thm33}
$E$ is classical if and only if $M(E)$ has one row.
\end{thm}
\begin{proof}
Suppose $E$ has atoms $\brac{a_1,a_2,\ldots ,a_m}$ and $M(E)=[x_1x_2\cdots x_m]$ has one row so $\bigoplus\limits _{i=1}^mx_ia_i=1$. Since $n(a_1)a_1$ is defined, we have by Theorem~\ref{thm21} that
\begin{equation*}
1=n(a_1)a_1\oplus [n(a_1)a_1]'=n(a_1)a_1\oplus y_2a_2\oplus\cdots\oplus y_na_n
\end{equation*}
Since $M(E)$ has one row, we conclude that $n(a_1)=x_1$. In a similar way, $n(a_i)=x_i$, $i=1,2,\ldots ,m$. Hence, $\oplus n(a_i)a_i=1$ so by Theorem~\ref{thm25}, $E$ is classical. Conversely, if $E$ is classical, then by Theorem~\ref{thm25}, we have $\oplus n(a_i)a_i=1$. Therefore, $M(E)$ has the row $[n(a_1)n(a_2)\cdots n(a_n)]$. If $M(E)$ has another row $[x_1x_2\cdots x_m]$, then $x_i<n(a_i)$ for some $i$. But then 
\begin{equation*}
\oplus x_ia_i<\oplus n(a_i)a_i=1
\end{equation*}
which is a contradiction. Hence, $M(E)$ has one row.
\end{proof}

We now list the classical effect algebras with 2 to 16 elements. In this table, the first column contains the number of elements $n$, the second column the number of atoms, the third gives the number of classical effect algbras with $n$ elements, the fourth gives the row representations and the fifth gives the factorizations of $n$. For example, the factorizations of 8 are $8,4\tbullet 2$ and $2\tbullet  2\tbullet  2$. After the table we show how these various quantities are obtained.
\vskip 2pc

{\hskip -3pc
\begin{tabular}{|c|c|c|c|c|c|}
\hline
Elements&Atoms&Classical&Rows&Factorizations\\
\hline
2&1&1&[1]&2\\
\hline
3&1&1&[2]&3\\
\hline
4&1,2&2&[3],[1\,1]&$4,2\tbullet 2$\\
\hline
5&1&1&[4]&5\\
\hline
6&1,2&2&[5],[1\,2]&$6,2\tbullet 3$\\
\hline
7&1&1&[6]&7\\
\hline
8&$1,2,3$&3&[7],[1\,3],[1\,1\,1]&$8,2\tbullet 4,2\tbullet 2\tbullet 2$\\
\hline
9&1,2&2&[8],[2\,2]&$9,3\tbullet 3$\\
\hline
10&1,2&2&[9],[1\,4]&$10,2\tbullet 5$\\
\hline
11&1&1&[10]&11\\
\hline
12&1,2,3&4&[11],[1\,5],[2\,3],[1\,1\,2]&$12,2\tbullet 6,3\tbullet 4,2\tbullet 2\tbullet 3$\\
\hline
13&1&1&[12]&13\\
\hline
14&1,2&2&[13],[1\,6]&$14,2\tbullet 7$\\
\hline
15&1,2&2&[14],[2\,4]&$15,3\tbullet 5$\\
\hline
16&$1,2,3,4$&5&[15],[3\,3],[1\,1\,1],[1\,1\,3],[1\,1\,1\,1]&$16,4\tbullet 4,2\tbullet 8,2\tbullet 2\tbullet 4,2\tbullet 2\tbullet 2\tbullet 2$\\
\hline
\noalign{\smallskip}
\end{tabular}}
\vskip 2pc

If we have $n$ effects in $E$, we consider all the factorizations of $n$. We always have the scale effect algebra $[n-1]$ with one atom. This corresponds to the factorization $n=n$. The factorizations with two factors $n=x_1\tbullet x_2$ corresponds to 2 atoms and classical effect algebras $[x_1-1\ x_2-1]$. The factorizations with three factors $n=x_1\tbullet x_2\tbullet x_3$ correspond to 3 atoms and classical effect algebras $[x_1-1\ x_2-1\ x_3-1]$. Continuing, the factorizations $n=x_1\tbullet x_2\cdots\tbullet  x_m$ correspond to $m$ atoms with classical effect algebras $[x_1-1\ x_2-1\cdots x_m-1]$. As a corollary to these observations, we have 

\noindent (1)\enspace There is exactly one classical effect algebra with $n$ elements if and only if $n$ is prime.\newline
(2)\enspace If $n$ has prime factorization $n=n_1n_2\cdots n_m$, then there exists a classical effect algebra with $n$ elements and $i$ atoms for $i=1,2,\ldots ,m$.

\begin{example}  
The integer 36 has a large number of factorizations
\begin{align*}
36&=2\tbullet 18=3\tbullet 12=4\tbullet 9=6\tbullet 6=2\tbullet 2\tbullet 9=4\tbullet 3\tbullet 3\\
   &=2\tbullet 3\tbullet 6=2\tbullet 2\tbullet 3\tbullet 3
\end{align*}
We conclude that there are nine classical effect algebras with 36 elements. One has 1 atom, four have 2 atoms, three have 3 atoms and one has 4 atoms. The row representations of these effect algebras are:
\begin{equation*}
[35],[1\ 17],[2\ 11],[3\ 8],[5\ 5],[1\ 1\ 18],[2\ 2\ 3],[1\ 2\ 5],[1\ 1\ 2\ 2]\qquad\square
\end{equation*}
\end{example}

We have stated that if $E$ is classical with $n$ elements and $n$ has factorization $n=n_1n_2\cdots n_m$, then $E$ has $m$ atoms and matrix representation $M=[n_1-1\ n_2-1\cdots n_m-1]$. We now justify this statement. First, since $n_i\ge 2$, all the rows and columns of $M$ have at least one nonzero entry so Condition~(1) of Theorem~\ref{thm32} holds. Since there is only one row in $M$, Condition~(2) and (3) of Theorem~\ref{thm32} hold. Hence, $M$ is a matrix representation of a classical effect algebra with $n$ elements and $m$ atoms. Conversely, if $M(E)=[n_1-1\ n_2-1\cdots n_m-1]$ where $n_i\ge 2$, $i=1,2,\ldots ,m$, then there are $m$ atoms $a_1,a_2,\ldots ,a_m$ in $E$. Since every element of $E$ has the form $\oplus x_ia_i$, $0\le x_i\le n_i-1$, we conclude that there are $n=n_1n_2\cdots , n_m$ elements in $E$.

\section{Sum Tables}  
We now consider sum tables for classical effect algebras \cite{gt24}. Although the matrix representation is a concise way of displaying an effect algbra $E$, it does not easily identify the sums which are the basic operations of $E$. This is accomplished using sum tables. Since the sums $a\oplus 0=a$ and $a\oplus 1$ is not defined for $a\ne 0$, we need not include these trivial sums in the table. We use the symbol $N$ to specify that $a\oplus b$ is not defined.

The sum table for the 3-element classical effect algebra $[2]=\brac{0,a,1}$ is
\medskip

\begin{tabular}{c|c|c|c}
$+$&$a$\\
\hline
$a$&1\\
\hline
\end{tabular}\,.
\medskip

\setlength{\parindent}{0pt} 

This shows that $2a=a\oplus a=1$ which is the only nontrivial sum. The sum table for the 4-element classical effect algebra $[3]=\brac{0,a,b,1}$ where $3a=1$ becomes
\medskip

\begin{tabular}{c|c|c|c|c|}
$+$&$a$&$b$\\
\hline
$a$&$b$&1\\
\hline
$b$&1&$N$\\
\hline
\end{tabular}
\medskip

This shows that $a\oplus a=b$, $a\oplus b=b\oplus a=1$ and $b\oplus b$ is not defined. The sum table for the 4-element classical effect algebra
$[1\, 1]=\brac{0,a,b,1}$ where $a\oplus b=1$ becomes
\medskip

\begin{tabular}{c|c|c|c}
$+$&$a$&$b$\\
\hline
$a$&$N$&1\\
\hline
$b$&1&$N$\\
\hline
\end{tabular}\,. 
\medskip

This shows that $a\oplus b=b\oplus a=1$ and $a\oplus a$, $b\oplus b$ are not defined.
We now display the sum tables for the classical effect algebras with 5 to 8 elements.

5-elements: $[4]=\brac{0,a,b,c,1}, 4a=1$
\medskip

\begin{tabular}{c|c|c|c|c}
$+$&$a$&$b$&$c$\\
\hline
$a$&$b$&$c$&1\\
\hline
$b$&$c$&1&$N$\\
\hline
$c$&1&$N$&$N$\\
\hline
\end{tabular}\,. 
\medskip

6-elements: $[5]=\brac{0,a,b,c,d,1}, 5a=1$
\medskip

\begin{tabular}{c|c|c|c|c|c}
$+$&$a$&$b$&$c$&$d$\\
\hline
$a$&$b$&$c$&$d$&1\\
\hline
$b$&$c$&$d$&1&$N$\\
\hline
$c$&$d$&1&$N$&$N$\\
\hline
$d$&1&$N$&$N$&$N$\\
\hline
\end{tabular}\,. 
\medskip

6-elements: $[1\ 2]=\brac{0,a,b,c,d,1}, a\oplus 2b=1$
\medskip

\begin{tabular}{c|c|c|c|c|c}
$+$&$a$&$b$&$c$&$d$\\
\hline
$a$&$N$&$c$&$N$&1\\
\hline
$b$&$c$&$d$&1&$N$\\
\hline
$c$&$N$&1&$N$&$N$\\
\hline
$d$&1&$N$&$N$&$N$\\
\hline
\end{tabular}\,. 
\medskip

7-elements: $[6]=\brac{0,a,b,c,d,e,1}, 6a=1$
\medskip

\begin{tabular}{c|c|c|c|c|c|c}
$+$&$a$&$b$&$c$&$d$&$e$\\
\hline
$a$&$b$&$c$&$d$&e&1\\
\hline
$b$&$c$&$d$&$e$&1&$N$\\
\hline
$c$&$d$&$e$&1&$N$&$N$\\
\hline
$d$&$e$&1&$N$&$N$&$N$\\
\hline
$e$&1&$N$&$N$&$N$&$N$\\
\hline
\end{tabular}\,. 
\bigskip

8-elements: $[7]=\brac{0,a,b,c,d,e,f,1}, 7a=1$
\medskip

\begin{tabular}{c|c|c|c|c|c|c|c}
$+$&$a$&$b$&$c$&$d$&$e$&$f$\\
\hline
$a$&$b$&$c$&$d$&e&$f$&1\\
\hline
$b$&$c$&$d$&$e$&$f$&1&$N$\\
\hline
$c$&$d$&$e$&$f$&1&$N$&$N$\\
\hline
$d$&$e$&$f$&1&$N$&$N$&$N$\\
\hline
$e$&$f$&1&$N$&$N$&$N$&$N$\\
\hline
$f$&1&$N$&$N$&$N$&$N$&$N$\\
\hline
\end{tabular}\,. 
\bigskip

8-elements: $[1\ 3]=\brac{0,a,b,c,d,e,f,1}, a\oplus 3b=1$
\medskip

\begin{tabular}{c|c|c|c|c|c|c|c}
$+$&$a$&$b$&$c$&$d$&$e$&$f$\\
\hline
$a$&$N$&$c$&$N$&e&$f$&1\\
\hline
$b$&$c$&$d$&$e$&$f$&1&$N$\\
\hline
$c$&$N$&$e$&$N$&1&$N$&$N$\\
\hline
$d$&$e$&$f$&1&$N$&$N$&$N$\\
\hline
$e$&$f$&1&$N$&$N$&$N$&$N$\\
\hline
$f$&1&$N$&$N$&$N$&$N$&$N$\\
\hline
\end{tabular}\,. 
\medskip

8-elements: $[1\ 1\ 1]=\brac{0,a,b,c,d,e,f,1}, a\oplus b\oplus c=1$
\medskip

\begin{tabular}{c|c|c|c|c|c|c|c}
$+$&$a$&$b$&$c$&$d$&$e$&$f$\\
\hline
$a$&$N$&$d$&$e$&N&$N$&1\\
\hline
$b$&$d$&$N$&$f$&$N$&1&$N$\\
\hline
$c$&$e$&$f$&$N$&1&$N$&$N$\\
\hline
$d$&$N$&$N$&1&$N$&$N$&$N$\\
\hline
$e$&$N$&1&$N$&$N$&$N$&$N$\\
\hline
$f$&1&$N$&$N$&$N$&$N$&$N$\\
\hline
\end{tabular}\,. 
\medskip

\begin{example}  
Of course, any finite effect algebra has a sum table. For example, the effect algebra $E=\brac{0,a,b,c,1}$ with matrix representation
$\begin{bmatrix}1&0&1\\0&2&0\end{bmatrix}$ so that $a\oplus c=1$, $2b=1$ has sum table
\medskip

\begin{tabular}{c|c|c|c|c|c|}
$+$&$a$&$b$&$c$\\
\hline
$a$&$N$&$N$&1\\
\hline
$b$&$N$&1&$N$\\
\hline
$c$&1&$N$&$N$\\
\hline
\end{tabular}
\medskip
\hfill\qedsymbol
\end{example}

\setlength{\parindent}{20pt}

\section{States on Classical Effect Algebras}  
A \textit{state} on an effect algebra $E$ is a function $s\colon E\to [0,1]\subseteq\real$ such that $s(1)=1$ and $s(a\oplus b)=s(a)+s(b)$ whenever $a\perp b$. A state describes the initial condition of a physical system corresponding to $E$ and we denote the set of states on $E$ by $\sscript (E)$. If $a\in E$, $s\in\sscript (E)$, then $s(a)$ is interpreted as the probability that $a$ occurs when $a$ is measured and the system is in state $s$. A set of states $S\subseteq\sscript (E)$ in \textit{order-determing} if $s(a)\le s(b)$ for all $s\in S$ implies that $a\le b$. If $s_i\in\sscript (E)$, $i=1,2,\ldots ,n$, and $\lambda _i\in [0,1]\le\real$ with $\sum\lambda _i=1$, then $\sum\lambda _is_i\in\sscript (E)$ is called a \textit{convex combination} of $\brac{s_1,s_2,\ldots ,s_n}$.

\begin{thm}    
\label{thm51}
{\rm{(i)}}\enspace If $E$ is a classical effect algebra with $m$ atoms, the there exists a set of order-determining states $S=\brac{s_1,s_2,\ldots ,s_m}$ on $E$ such that every state $s\in\sscript (E)$ is a convex combination of states in $S$.
{\rm{(ii)}}\enspace A classical effect algebra is a quantum effect algebra.
\end{thm}
\begin{proof}
(i)\enspace Since $E$ has $m$ atoms $a_1,a_2,\ldots ,a_m$ we have $M(E)=[x_1-1\ x_2-1\ldots x_m-1]$ where $x_i\in\Natural$, $x_i\ge 2$ and $\bigoplus\limits _{i=1}^m(x_i-1)a_i=1$. Also, every $a\in E$ has the form
\begin{equation}                
\label{eq51}
a=\bigoplus _{i=1}^my_ia_i,\quad y_i\in\Natural _0,\quad 0\le y_i\le x_i-1
\end{equation}
Let $s_i\colon\brac{a_1,a_2,\ldots ,a_m}\to [0,1]\subseteq\real$ be the function $s_i(a_j)=\tfrac{1}{x_i-1}\,\delta _{ij}$,  where $\delta _{ij}$ is the Kronecker delta. If $a\in E$ has the form \eqref{eq51}, define 
\begin{equation*}
s_i(a)=\sum _{j=1}^my_is_i (a_j)=\tfrac{y_i}{x_i-1}\,,\quad i=1,2,\ldots ,m
\end{equation*}
We then have
\begin{equation*}
s_i(1)=\sum _{j=1}^m(x_i-1)s_i(a_j)=\tfrac{x_i-1}{x_i-1}=1,\quad i=1,2,\ldots ,m
\end{equation*}
If $a$ has the form \eqref{eq51} and $b=\bigoplus\limits _{i=1}^my'_ia_i$ with $a\perp b$, then
$a\oplus b=\bigoplus\limits _{i=1}^m(y_i+y'_i)a_i$ and 
\begin{equation*}
s_i(a\oplus b)=\tfrac{(y_i+y'_i)}{x_i-1}=s_i(a)+s_i(b),\quad i=1,2,\ldots ,m
\end{equation*}
Hence, $s_i\colon E\to [0,1]$ is a state, $i=1,2,\ldots ,m$. If $s_i(a)\le s_i(b)$, $i=1,2,\ldots ,m$ and $a,b$ have the previous forms, then
$\tfrac{y_i}{x_i-1}\le\tfrac{y'_i}{x_i-1}$ so $y_i\le y'_i$, $i=1,2,\ldots ,m$. It follows that $S=\brac{s_1,s_2,\ldots ,s_m}$ is an order-determining set of states. We now show that every $s\in\sscript (E)$ is a convex combination of states in $S$. Let $s\in\sscript (E)$ and define
$\lambda _i=(x_i-1)s(a_i)$, $i=1,2,\ldots ,m$. Then
\begin{equation*}
\sum _{i=1}^m\lambda _i=\sum _{i=1}^m(x_i-1)s(a_i)=s\sqbrac{\bigoplus _{i=1}^m(x_i-1)a_i}=s(1)=1
\end{equation*}
If $a$ has the form \eqref{eq51}, then
\begin{equation*}
s(a)=\sum _{i=1}^my_is(a_i)=\sum _{i=1}^my_i\tfrac{\lambda _i}{x_i-1}=\sum _{i=1}^m\lambda _is_i(a)
   =\paren{\sum _{i=1}^m\lambda _is_i}(a)
\end{equation*}
Therefore, $s=\sum\limits _{i=1}^m\lambda _is_i$ so $s$ is a convex combination of states in $S$.
(ii)\enspace It is shown in \cite{gt24} that a finite effect algebra is a quantum effect algebra if and only if it has an order-determining set of states. Applying (i) it follows that a classical effect algebra is quantum.
\end{proof}

More details concerning quantum effect algebras are given in \cite{gt24}. The following table gives the number of classical and quantum effect algebras with 2 to 6 elements.
\bigskip

{\hskip -3pc
\begin{tabular}{|c|c|c|c|c|c|c|}
\hline
Number of&Number of&Classical&Quantum&Non-Quantum\\
Elements&Effect Algebras&&Not Classical&\\
\hline
2&1&1&0&0\\
\hline
3&1&1&0&0\\
\hline
4&3&2&0&1\\
\hline
5&4&1&1&2\\
\hline
6&10&2&2&6\\
\hline
\noalign{\smallskip}
\end{tabular}}
\vskip 1pc

We have already displayed the classical effect algebras with 2 to 6 elements. We now display the others. The non-quantum effect algebra with 4 elements is $\begin{bmatrix}2&0\\0&2\end{bmatrix}$
\medskip

\noindent The quantum but not classical effect algebra with 5 elements has the sum table and matrix representation
\medskip

\begin{tabular}{c|c|c|c|c|c|}
$+$&$a$&$b$&$c$\\
\hline
$a$&$N$&1&$N$\\
\hline
$b$&1&$N$&$N$\\
\hline
$c$&$N$&$N$&1\\
\hline
\end{tabular}\qquad
$\begin{bmatrix}1&1&0\\0&0&2\end{bmatrix}$ 
\medskip

\noindent The two quantum but not classical effect algebras with 6 elements have the sum tables and matrix representations
\medskip

\noindent
\begin{tabular}{c|c|c|c|c|c}
$+$&$a$&$b$&$c$&$d$\\
\hline
$a$&$N$&1&$N$&$N$\\
\hline
$b$&1&$N$&$N$&$N$\\
\hline
$c$&$N$&$N$&$N$&1\\
\hline
$d$&$N$&$N$&1&$N$\\
\hline
\end{tabular}
$\begin{bmatrix}1&1&0&0\\0&0&1&1\end{bmatrix}$
\begin{tabular}{c|c|c|c|c|c}
$+$&$a$&$b$&$c$&$d$\\
\hline
$a$&$b$&1&$N$&$N$\\
\hline
$b$&1&$N$&$N$&$N$\\
\hline
$c$&$N$&$N$&$N$&1\\
\hline
$d$&$N$&$N$&1&$N$\\
\hline
\end{tabular}
$\begin{bmatrix}3&\!0&\!0\\0&\!1&\!1\end{bmatrix}$ atoms $a,c,d$
\medskip

\noindent The six non-quantum effect algebras with 6 elements have matrix representations
\medskip

\noindent
$\begin{bmatrix}1&2\\3&0\end{bmatrix}$,\ $\begin{bmatrix}3&0\\0&3\end{bmatrix}$,\ $\begin{bmatrix}2&0\\0&4\end{bmatrix}$,\ 
$\begin{bmatrix}3&0&0\\0&2&0\\0&0&2\end{bmatrix}$,\ 
$\begin{bmatrix}1&1&0&0\\0&0&2&0\\0&0&0&2\end{bmatrix}$,\ 
$\begin{bmatrix}2&0&0&0\\0&2&0&0\\0&0&2&0\\0&0&0&2\end{bmatrix}$
\bigskip

\section{Composite Effect Algebras}  
If $E$ and $F$ are effect algebras, consider their cartesian product $E\times F=\brac{(a,b)\colon a\in E,b\in F}$. It is shown \cite{gt24} that $E\times F$ is an effect algebra if we define $0=(0,0)$, $1=(1,1)$, $(a,b)'=(a',b')$, $(a_1,b_1)\perp (a_2,b_2)$ when $a_1\perp a_2$,
$b_1\perp b_2$ and in this
\begin{equation*}
(a_1,b_1)\oplus (a_2,b_2)=(a_1\oplus a_2,b_1\oplus b_2)
\end{equation*}
We then call $E\times F$ the \textit{composite} of $E$ and $F$. We think of the composite $E\times F$ as the effect algebra describing a combination of two physical systems whose effect algebras are $E$ and $F$, respectively. It is clear that if $E$ has $m$ elements and $F$ has $n$ elements, then $E\times F$ has $mn$ elements.

\begin{thm}    
\label{thm61}
{\rm{(i)}}\enspace For $(a_1,b_1),(a_2,b_2)\in E\times F$, we have $(a_1,b_1)\le (a_2,b_2)$ if and only if $a_1\le a_2$, $b_1\le b_2$.
{\rm{(ii)}}\enspace An effect $d\in E\times F$ is an atom if and only if $d=(a,0)$ or $d=(0,b)$ where $a$ is an atom of $E$ and $b$ is an atom of $F$.
\end{thm}
\begin{proof}
(i)\enspace If $a_1\le a_2$, $b_1\le b_2$ there exists $c\in E$, $d\in F$ such that $a_2=a_1\oplus c$ and $b_2=b_1\oplus d$. Then
\begin{equation*}
(a_2,b_2)=(a_1\oplus c,b_1\oplus d)=(a_1,b_1)\oplus (c,d)
\end{equation*}
so $(a_1,b_1)\le (a_2,b_2)$. Conversely, if $(a_1,b_1)\le (a_2,b_2)$ then there exists $(c,d)\in E\times F$ such that
\begin{equation*}
(a_2,b_2)=(a_1,b_1)\oplus (c,d)=(a_1\oplus c,b_1\oplus d)
\end{equation*}
Hence, $a_2=a_1\oplus c$, $b_2=b_1\oplus d$ so $a_1\le a_2$, $b_1\le b_2$.
(ii)\enspace If $d=(a,0)$ where $a$ is an atom and $e=(e_1,e_2)\le d$, then by (i) $e_1\le a$ and $e_2=0$. Hence, $e_1=0$ or $e_1=a$ so $e=(0,0)$ or $e=(a,0)=d$. Thus, $d$ is an atom. If $d=(0,b)$ where $b$ is an atom, then in a similar way $d$ is an atom. Conversely, suppose $d=(a,b)$ is an atom in $E\times F$ where $a\ne 0$. If $e\le a$, then by (i) $(e,b)\le (a,b)$ so $(e,b)=0$ or $(e,b)=(a,b)$. If
$(e,b)=0$ then $e=0$ and otherwise $e=a$. Therefore, $a$ is an atom. Now if $b\ne 0$, then $0\ne (a,0)\le (a,b)$ but $(a,0)\ne (a,b)$ which contradicts the fact that $(a,b)$ is an atom. Hence, $b=0$ and $d=(a,0)$ where $a$ is an atom. In a similar way, if $(a,b)$ is an atom where $b\ne 0$ then $a=0$ and $b$ is an atom.
\end{proof}

Let $A\in\rmob(E)$, $B\in\rmob (F)$ with $A=\brac{a_1,a_2,\ldots ,a_m}$, $B=\brac{b_1,b_2,\ldots ,b_n}$. We define the 
\textit{product observable} $A\times B\in\rmob (E\times F)$ by
\begin{equation*}
A\times B=\brac{(a_1,0),(a_2,0),\ldots ,(a_m,0),(0,b_1),(0,b_2),\ldots ,(0,b_n)}
\end{equation*}
We see that $A\times B$ has $m+n$ effects. An observable $C\in\rmob (E\times F)$ is called a \textit{mixed} observable if it is not a product observable. We can combine $A$ and $B$ in various ways to obtain many mixed observables. For example, if $m\le n$, then
$C\in\rmob (E\times F)$ given by
\begin{equation}                
\label{eq61}
C=\brac{(a_1,b_1),(a_2,b_2),\ldots ,(a_m.b_m),(0,b_{m+1}),\ldots ,(0,b_n)}
\end{equation}
is a mixed observable. To show that $C$ is indeed an observable, the sum of its effects is 
\begin{align*}
&(a_1,b_1)\oplus (a_2,b_2)\oplus\cdots\oplus (a_m,b_m)\oplus (0,b_{m+1})\oplus\cdots\oplus (0,b_n)\\
   &=(a_1\oplus a_2\oplus\cdots\oplus a_m,b_1\oplus b_2\oplus\cdots\oplus b_m)\oplus (0,b_{m+1}\oplus b_{m+2}\oplus\cdots\oplus b_n)\\
   &=(1,b_1\oplus b_2\oplus\cdots\oplus b_n)=(1,1)
\end{align*}
We can replace the $(a_i,b_i)$ effects in $C$ with effects $(a_{i_j},b_{i_k})$ to obtain other mixed observables. Notice that $C$ has only $n$ effects. We now show that if $C\in\rmob (E\times F)$, then there exists a product observable $A\times B$, $A\in\rmob(E)$, $B\in\rmob (F)$ such that $C=A\times B$. We can assume that $C$ has the form \eqref{eq61}. Let $A=\brac{a_1,\ldots ,a_m,0,0,\ldots ,0}$ where there are
$(n-m)$ 0's and let $B=\brac{b_1,b_2,\ldots ,b_n}$. Then
\begin{equation*}
A\times B=\brac{(a_1,0),(a_2,0),\ldots ,(a_m,0), (0,0),\ldots ,(0,0),(0,b_1),(0,b_2),\ldots ,(0,b_n)}
\end{equation*}
and we have $C=A\times B$.

\begin{thm}    
\label{thm62}
Two effect algebras $E,F$ are classical if and only if $E\times F$ is classical.
\end{thm}
\begin{proof}
If $E,F$ are classical, there exist generating observables $A\in\rmob (E)$, $B\in\rmob (F)$. Let $C=A\times B\in\rmob (E\times F)$. Since any $(a,b)\in E\times F$ has the form $(a,b)=(a,0)\oplus (0,b)$ it is clear that $C$ is a generating observable for $E\times F$. Hence,
$E\times F$ is classical. Conversely, suppose $E\times F$ is classical with generating observable $C=\brac{(a_1,b_1),\ldots ,(a_n,b_n)}$. Letting $A=\brac{a_1,a_2,\ldots ,a_n}$, $B=\brac{b_1,b_2,\ldots ,b_n}$ we have 
\begin{equation*}
(a_1\oplus a_2\oplus\cdots\oplus a_n,b_1\oplus b_2\oplus\cdots\oplus b_n)=(a_1,b_1)\oplus (a_2,b_2)\oplus\cdots\oplus (a_n,b_n)=1
\end{equation*}
Hence, $a_1\oplus a_2\oplus\cdots\oplus a_n=1$, $b_1\oplus b_2\oplus\cdots\oplus b_n=1$ so $A\in\rmob (E)$, $B\in\rmob (F)$. If 
$a\in E$, since $C$ is generating
\begin{equation*}
(a,0)=(a_{i_1},b_{i_1})\oplus (a_{i_2},b_{i_2})\oplus\cdots\oplus (a_{j_r}, b_{j_r})
  =(a_{i_1}\oplus a_{i_2}\oplus\cdots\oplus a_{j_r},b_{i_1}\oplus b_{i_2}\oplus\cdots\oplus b_{j_r})
\end{equation*}
Therefore, $a=a_{i_1}\oplus a_{i_2}\oplus\cdots\oplus a_{j_r}$ so $A$ is generating. Similarly, $B$ is generating. It follows that $E,F$ are classical.
\end{proof}

It is shown in \cite{gt24} that $E,F$ are quantum if and only if $E\times F$ in quantum. We next consider matrix representations of composite effect algebras. Let $E,F$ be effect algebras with $m,t$ atoms, respectively. Their matrix representations become
\begin{equation*}
M(E)=\begin{bmatrix}x_{11}&x_{12}&\cdots&x_{1m}\\\vdots&\vdots&&\vdots\\x_{n1}&x_{n2}&\cdots&x_{nm}\end{bmatrix}_{n\times m}\quad
M(F)=\begin{bmatrix}y_{11}&y_{12}&\cdots&y_{1t}\\\vdots&\vdots&&\vdots\\y_{s1}&y_{s2}&\cdots&y_{st}\end{bmatrix}_{s\times t}
\end{equation*}
Letting $a_1,a_2,\ldots ,a_m$ be the atoms of $E$ and $b_1,b_2,\cdots ,b_t$ be the atoms of $F$, we have
$\bigoplus\limits _{j=1}^mx_{ij}a_j=1$, $i=1,2,\ldots ,n$, $\bigoplus\limits _{k=1}^ty_{ik}b_k=1$, $i=1,2,\ldots ,s$. Hence,
\begin{equation*}
\bigoplus _{j=1}^mx_{ij}(a_j,0)\bigoplus _{k=1}^ty_{ik}(0,b_k)=(1,0)\oplus (0,1)=1
\end{equation*}
so $E\times F$ has $m+t$ atoms and $M(E\times F)$ has $m+t$ columns and $ns$ rows. We conclude that $M(E\times F)$ has the form:
\begin{equation*}
M(E\times F)=\begin{bmatrix}x_{11}&x_{12}&\cdots&x_{1m}&y_{11}&y_{12}&\cdots&y_{1t}\\
  \vdots&&&&&&&\\
  x_{11}&x_{12}&\cdots&x_{1m}&y_{s1}&y_{s2}&\cdots&y_{st}\\
  x_{21}&x_{22}&\cdots&x_{2m}&y_{11}&y_{12}&\cdots&y_{1t}\\
  \vdots&&&&&&&\\
  x_{21}&x_{22}&\cdots&x_{2m}&y_{s1}&y_{s2}&\cdots&y_{st}\\
  \vdots&&&&&&&\\
  x_{m1}&x_{m2}&\cdots&x_{nm}&y_{11}&y_{12}&\cdots&y_{1t}\\
  \vdots&&&&&&&\\
  x_{m1}&x_{m2}&\cdots&x_{nm}&y_{s1}&y_{s2}&\cdots&y_{st}\\
\end{bmatrix}_{ns\times (m+t)}
\end{equation*}
The situation is much simpler for classical effect algebras. In this case, if $E$ has $m$ atoms and $F$ has $t$ atoms, then by Theorem~\ref{thm62}, $E$ has one row and $m+t$ columns. Hence, if $M(E)=[x_1x_2\cdots x_m]$ and $M(F)=[y_1y_2\cdots y_t]$ then 
\begin{equation*}
M[E\times F]=[x_1x_2\cdots x_my_1y_2\cdots y_t]
\end{equation*}

\begin{example}  
We find the matrix representations for composites of small classical effect algebras.

$E,F$ have 4 elements: $M(E)=[3]$, $M(F)=[3]$

$E\times F$ has 16 elements: $M(E\times F)=[3\ 3]$

$E,F$ have 4 elements: $M(E)=[3]$, $M(F)=[1\ 1]$

$E\times F$ has 16 elements: $M(E\times F)=[1\ 1\ 3]$

$E,F$ have 4 elements: $M(E)=[1\ 1]$, $M(F)=[1\ 1]$

$E\times F$ has 16 elements: $M(E\times F)=[1\ 1\ 1\ 1]$

$E$ has 2 elements, $F$ has 8 elements: $M(E)=[1]$, $M(F)=[7]$

$E\times F$ has 16 elements: $M(E\times F)=[1\ 7]$\hfill\qedsymbol
\end{example}

\begin{example}  
We find the matrix representations for composites of small nonclassical effect algebras.

$E$ has 2 elements, $F$ has 4 elements: $M(E)=[1]$, $M(F)=\begin{bmatrix}2&0\\0&2\end{bmatrix}$

$E\times F$ has 8 elements: $M(E)=\begin{bmatrix}1&2&0\\1&0&2\end{bmatrix}$
\smallskip

$E,F$ have 4 elements: $M(E)=\begin{bmatrix}2&0\\0&2\end{bmatrix}$, $M(F)=\begin{bmatrix}2&0\\0&2\end{bmatrix}$
\smallskip

$E\times F$ has 16 elements: $M(E\times F)=\begin{bmatrix}2&0&2&0\\2&0&0&2\\0&2&2&0\\0&2&0&2\end{bmatrix}$
\hfill\qedsymbol
\end{example}


\begin{thebibliography}{99}
\bibitem{bkz23}G.\,Binezak, J.\,Kaleta, and A.\,Zembrzuski, Matrix representations of finite effect algebras, \textit{Kybernetika} \textbf{59}, 737--751 (2023).
\bibitem{bgl95}P.\,Busch, M.\,Grabowski and P.\,Lahti, \textit{Operational Quantum Physics}, Springer-Verlag, Berlin, 1995.
\bibitem{blm96}P.\,Busch, P.\,Lahti and P.\,Mittlestaedt, \textit{The Quantum Theory of Measurement}, Springer-Verlag, Berlin, 1996.
\bibitem{dp00}A.\,Dvure\v censkij and S.\,Pulmannov\'a, \textit{New Trends in Quantum Structures}, Kluwer Academic Publ., Bratislava, 2000.
\bibitem{fb94}D. Foulis and M. K. Bennett, Effect algebras and unsharp quantum logics, \textit{Found. Phys.} \textbf{24}, 1331-1352 (1994).
\bibitem{gud97}S.\,Gudder, Effect test spaces and effect algebras, \textit{Found.~Phys.} \textbf{27}, 287-304 (1997).
\bibitem{gt24}S.\,Gudder and T.\,Heinosaari, Finite (quantum) effect algebras, arXiv: quant-ph 2406.13775 (2024).
\end{thebibliography}
\end{document}